\documentclass[DIV8,10pt,a4paper,headings=small]{scrartcl}
\setkomafont{pagehead}{\small} 
\usepackage[breaklinks]{hyperref}
\usepackage{breakurl}
\usepackage{graphics,graphicx}
\usepackage{times}
\usepackage{amsmath,amsfonts,amssymb}
\usepackage{amsthm}
\usepackage[square,sort&compress,comma,numbers]{natbib}
\usepackage[english]{babel}
\usepackage{bbm}
\usepackage[usenames,dvipsnames]{xcolor}
\usepackage{setspace}
\usepackage{pgf,tikz}
\usetikzlibrary{calc}
\usetikzlibrary{positioning}
\usetikzlibrary{decorations}
\usetikzlibrary{decorations.pathreplacing}

\usepackage{ifthen}
\usepackage{mathtools}

\usepackage{enumitem}
\setlist{nolistsep,itemsep=\parskip,labelindent=\parindent,leftmargin=0pt,font=\bfseries}

\newcommand{\dd}[1]{\mathop{\mathrm{d}#1}}
\DeclareMathOperator{\Sym}{\mathrm{Sym}}

\newcommand{\e}{\mathrm{e}}

\DeclareMathOperator{\Tr}{\mathrm{Tr}} 
\DeclareMathOperator{\Perm}{\mathrm{Perm}}
\DeclareMathOperator{\Erfc}{\mathrm{Erfc}}
\DeclareMathOperator{\Erf}{\mathrm{Erf}}

\DeclareMathOperator{\supp}{\mathrm{supp}}

\DeclareMathOperator{\poly}{\mathrm{poly}}

\newcommand{\mc}[1]{\mathcal{#1}}
\newcommand{\mcA}{\mc{A}}
\newcommand{\mcH}{\mc{H}}
\newcommand{\mcB}{\mc{B}}
\newcommand{\mcS}{\mc{S}}
\newcommand{\mcP}{\mc{P}}
\newcommand{\mcQ}{\mc{Q}}
\newcommand{\mcD}{\mc{D}}

\newcommand{\mb}[1]{\mathbb{#1}}
\newcommand{\N}{\mb{N}}
\newcommand{\Z}{\mb{Z}}
\newcommand{\R}{\mb{R}}
\newcommand{\C}{\mb{C}}
\newcommand{\1}{\mb{I}}
\DeclareMathOperator*{\E}{\mathbb{E}}

\newcommand{\ket}[1]{\left|{#1}\right\rangle}

\newcommand{\bra}[1]{\left\langle{#1}\right|}

\DeclareMathOperator{\landauO}{O}

\renewcommand{\k}{N}

\newcommand{\Sr}[3][]{
  \ifthenelse{\equal{#1}{}}
    {S(#2\|#3)}
    {S_{#1}(#2\|#3)}
}

\newtheorem{theorem}{Theorem}
\newtheorem{lemma}[theorem]{Lemma}
\newtheorem{definition}[theorem]{Definition}
\newtheorem{observation}[theorem]{Observation}

\makeatletter
\newcommand{\affiliation}[1]{\def\@affiliation {#1}}
\renewenvironment{abstract}{\small}{}
\renewcommand{\maketitle}{%
  \thispagestyle{empty}

  \vspace*{2cm}

  {\onehalfspacing\Large\usekomafont{title}\noindent\@title\par}
  \vspace{4mm}

  {\usekomafont{title}\large \noindent\@author}\par
  \vspace{2mm}
  {\small \noindent\@affiliation}\par
  \vspace{2mm}
  {\small \noindent\@date}\par
  \vspace{4mm}
}

\begin{document}

\pdfinfo{ 
	/Author (C.\ Gogolin, M.\ Kliesch, L.\ Aolita, and J.\ Eisert) 
	/Title (Boson-Sampling in the light of sample complexity) 
	/Subject (03.67.-a, 89.70.Eg, 05.30.Jp, 42.50.-p)
	/Keywords (boson-sampling, sample complexity, quantum simulation, certification, verification, state discrimination, birthday paradox, linear quantum optics)
	} 


\title{Boson-Sampling in the light of sample complexity}
\author{C.\ Gogolin, M.\ Kliesch, L.\ Aolita, and J.\ Eisert}
\affiliation{Dahlem Center for Complex Quantum Systems, Freie Universit{\"a}t Berlin, 14195 Berlin, Germany}
\date{September 16, 2013}

\maketitle

\begin{abstract}
  Boson-Sampling is a classically computationally hard problem that can --- in principle --- be efficiently solved with quantum linear optical networks.
  Very recently, a rush of experimental activity has ignited with the aim of developing such devices as feasible instances of quantum simulators.
  Even approximate Boson-Sampling is believed to be hard with high probability if the unitary describing the optical network is drawn from the Haar measure.
  In this work we show that in this setup, with probability exponentially close to one in the number of bosons, no symmetric algorithm can distinguish the Boson-Sampling distribution from the uniform one from fewer than exponentially many samples.
  This means that the two distributions are operationally indistinguishable without detailed a priori knowledge.
  We carefully discuss the prospects of efficiently using knowledge about the implemented unitary for devising non-symmetric algorithms that could potentially improve upon this.
  We conclude that due to the very fact that Boson-Sampling is believed to be hard, efficient classical certification of Boson-Sampling devices seems to be out of reach.
\end{abstract}

\section{Introduction}
\label{sec:introduction}
Quantum information theory suggests that it should be possible to design physical devices performing information processing tasks that cannot be classically efficiently simulated.
The most spectacular example of this type known to date is a fully-fletched Shor-class quantum computer, able to factorize numbers efficiently, hence solving a practically relevant problem for which no classical efficient algorithm is known \cite{Shor}.
Needless to say, the actual physical realisation of such a device is extraordinarily difficult for a number of reasons, the difficulty of protecting quantum systems from the unwanted effects of decoherence being only one of them.
In the light of this observation, it has become a very important milestone to identify devices that can solve some problem that seems impossible to be realised classically, or --- in the wording of a blog entry \cite{Preskill} ---  to achieve ``quantum supremacy''.
This is a challenge equally interesting for experimentalists as well as for theorists: On one hand, it surely is still very difficult to achieve the necessary degree of control, on the other hand, it is a challenge for complexity theorists and theoretical computer scientists to show that a task at hand is computationally hard.

A seminal theoretical step in this direction has recently been achieved with the introduction of the \emph{Boson-Sampling problem} \cite{Aaronson}. 
In this problem, the task is the following:
Given as input the unitary $U$, the number of modes $m$, and the number of photons $n$, together describing a quantum linear optical device (see Fig.~\ref{fig:bosonsamplingproblem}), sample from the output distribution of this device.
Ref.~\cite{Aaronson} establishes strong reasons to believe that classically sampling from this distribution up to a small error in 1-norm is computationally hard with high probability if the unitary $U$ is chosen from the Haar measure and $m$ is scaled suitably with $n$.
The hardness proof rests on the fact that approximating the probabilities of individual outcomes of such a device basically amounts to approximating the permanent of a submatrix of $U$ \cite{Scheel}, which, by a plausible complexity theoretic conjecture, is believed to be $\#P$ hard.
The main result of Ref.~\cite{Aaronson} suggests that a 1-norm approximate efficient classical simulation of Boson-Sampling would imply a collapse of the polynomial hierarchy to the third level (compare also Ref.~\cite{Bremner}).
This has triggered a rush of exciting experimental activity \cite{E1,E2,E3,E4}, aiming at realizing instances of Boson-Sampling, accompanied by theoretical discussions about what errors one should expect in such quantum linear optical experiments \cite{Rohde}.

In view of all this, a crucial question that arises is how to certify that a given experiment does actually solve the desired sampling problem, and how many repetitions of the experiment, i.e., samples from its output distribution, are needed for the certification.
In contrast to a machine that is efficiently factoring large numbers and hence solves a problem in NP, i.e., produces an output that can be checked efficiently on a classical computer, no efficient certification scheme for Boson-Sampling is known and it is not clear whether such a scheme can exist.

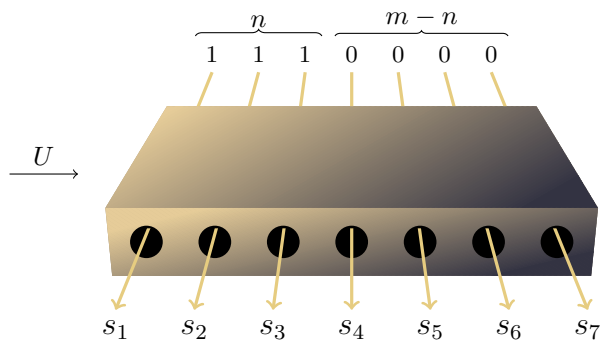
\begin{figure}[t]
  \centering 
  \begin{tikzpicture}[scale=0.9]
  \definecolor{box1}{rgb}{.9,.8,.6}
  \definecolor{box2}{rgb}{.2,.2,.26}
  \definecolor{bosonline}{rgb}{0.9,.8,.5}
  \def\xsizefront{3.6}
  \def\xsizeback{2.7}
  \def\xsizefrontbottom{3.5}
    \begin{scope}[top color=box1,bottom color=box2]
      \fill[box1!70!box2] (-\xsizefront,0.5) -- (\xsizefront,0.5) -- (\xsizeback,2) -- (-\xsizeback,2) ;
      \fill[box1!50!box2] (-\xsizefront,0.5) -- (\xsizefront,0.5) -- (\xsizefrontbottom,-0.5) -- (-\xsizefrontbottom,-0.5) ;      
      \shade[shading=axis,shading angle=70] (-\xsizefront,0.5) -- (\xsizefront,0.5) -- (\xsizeback,2) -- (-\xsizeback,2) ;
      \shade[shading=axis,shading angle=60] (-\xsizefront,0.5) -- (\xsizefront,0.5) -- (\xsizefrontbottom,-0.5) -- (-\xsizefrontbottom,-0.5) ;      
    \end{scope}
        \foreach \x/\j/\s in {-3/1/1,-2/2/1,-1/3/1,0/4/0,1/5/0,2/6/0,3/7/0} {
        \node[fill,circle,minimum size=0.43cm] at (\x,0) {};
        \begin{scope}
          \path[clip] (4,-2) -- (4,0.2) -- (\xsizefrontbottom,0.2) -- (-\xsizefrontbottom,0.2) -- (-\xsizefront,0.5) -- (-\xsizeback,2) -- (\xsizeback,2) -- (\xsizefront,0.5) -- (\xsizefrontbottom,0.2) -- (4,0.2) -- (4,3) -- (-4,3) -- (-4,-2);
          \draw[->,bosonline,very thick] (\x*.68,2.5) -- (\x*1.15,-1) node[at start,anchor=south,black] (ss\j)  {{\small $\s$}} node[at end,anchor=north,black] (s\j) {{\large $s_\j$}};
        \end{scope}
      }
      \draw[decorate,decoration=brace] (ss1.north west) -- (ss3.north east) node[midway,above] {$n$};
      \draw[decorate,decoration={brace,amplitude=2pt}] (ss4.north west) -- (ss7.north east) node[midway,above] {$m-n$};
      \draw[->] (-5,1) -- (-4,1) node[midway,above] {$U$};
    \end{tikzpicture}

    \caption{\label{fig:bosonsamplingproblem}The Boson-Sampling quantum device receives as input a unitary $U$, applies it to $m$ bosonic modes initialized with exactly one boson in each of the first $n$ modes and the vacuum in the remaining $m-n$ modes and outputs the results $(s_1,\dots, s_m)$ of local boson number measurements.
      The Boson-Sampling problem is to sample from the output distribution of this device where $U$ is fixed and chosen in the beginning from the Haar measure on $U(m)$.}
\end{figure}

As a first step towards a better understanding of the difficulty of certifying sampling devices we consider the task of deciding whether a device outputs samples from some given interesting probability distribution, for example the Boson-Sampling distribution, or the uniform one.
We approach this decision problem in two complementary settings that differ in the amount of information the certifier is allowed to use besides the samples output by the device.

\begin{description}
\item[State discrimination] (known probability distributions): 
  The certifier works under the assumption that the sampling device outputs independent identically distributed samples and that the output distribution is either of two completely known distributions (for example one of them being the uniform one), i.e., he has access to all the probabilities.
\item[Black box setting] (unknown probability distribution):
  The certifier works under the assumption that the sampling device outputs independent identically distributed samples, but has no a priori knowledge about the output distribution apart from the sample space.
\end{description}

We show that known bounds on the sample complexity imply that in the first setting a number of samples scaling polynomially 
with the number of bosons is sufficient to guarantee distinguishability of the Boson-Sampling distribution from the uniform one.
Notably, this does not imply that polynomially many samples are sufficient to certify that a given device samples from the correct distribution, even if unlimited computational power and full a priori knowledge about $m$, $n$, and $U$ implemented by the supposed device are assumed.
The reason for this is simply that being able to decide which of two given distributions a device samples from under the promise that it does indeed sample from either of the two, does not necessarily imply that one is also able to exclude that the device samples from any distribution outside of a small region around the target distribution from the same number of samples.

The complexity theoretic conjecture under which Boson-Sampling is a hard sampling problem, namely that it is expected to be $\#P$ hard to approximate the permanent, implies that approximating the probabilities of the individual outputs of a Boson-Sampling device is also computationally hard.
A classical certifier with limited computational power will hence have only very limited knowledge about the ideal output distribution of a supposed Boson-Sampling device.
The state discrimination scenario is thus far from realistically capturing the challenge of classically efficiently certifying a real Boson-Sampling device in the laboratory.

The realistic situation much more closely resembles the black box setting (we will discuss this in more detail in Section~\ref{sec:bosonsamplinginthecontextofsamplecomplexity}).
In this setting the certifier has no a priori knowledge about the output distribution.
It is hence reasonable to demand that its decision should only depend on how frequent the outcomes appear.
That is to say, knowing nothing about the probability distribution, the labels of the collected samples don't mean anything to the certifier, hence they should not influence his decision.

We formalize this in the notion of \emph{symmetric probabilistic decision algorithms} and show that such algorithms can give meaningful outputs only if they receive sufficiently many samples.
This is true not only for the task of distinguishing some distribution from the uniform one, but also in more general settings involving multiple sampling devices.
The number of samples necessarily depends on how \emph{flat} the distribution(s) are.
We call a probability distribution over a finite samples space \emph{$\epsilon$-flat} if the probability of the most likely outcome is upper bounded by $\epsilon$ or, equivalently, if its min entropy is larger than $\log_2 1/\epsilon$.
More precisely, we show that the output distribution of any symmetric probabilistic algorithm receiving at most $\landauO((1/\epsilon)^{1/4})$ samples from each of $\k$ sampling devices is with probability $1- \landauO(\k^2 \sqrt{\epsilon})$ independent of the distributions of the sampling devices if the distributions are $\epsilon$-flat (Theorem~\ref{theorem:symmetricalgorithmsandflatdistributions}).

We then show that the Boson-Sampling distribution is, for the interesting parameter regions and if $U$ is chosen from the Haar measure, with overwhelmingly high probability exponentially flat (Theorem~\ref{theorem:bosonsamplingdistributionisflat}).
Together, our findings imply that in the black box setting distinguishing the Boson-Sampling distribution from the uniform one requires exponentially many samples.

We emphasize that our analysis applies to the ideal situation without any experimental imperfections.
It is important to mention that our results concerning the flatness, just like the hardness proof of Ref.~\cite{Aaronson}, is probabilistic, in the sense of holding with an extremely high probability if the unitary describing the optical circuit is randomly chosen from the Haar measure and if $n$ and $m$ are sufficiently large.
To end up with, we identify a class of imperfect linear optical experimental situations for which one can classically efficiently sample from the output distribution even up to a constant small error in 1-norm.

The rest of this work is organized as follows.
First, in Section~\ref{sec:setting} we introduce the notation and recapitulate the setting considered in the Boson-Sampling problem.
Next, in Section~\ref{sec:bosonsamplinginthecontextofsamplecomplexity} we connect the problem of certifying a sampling devices with the decision problem of distinguishing its output distribution from the uniform distribution, explain the state discrimination and the black box setting in more detail, give upper and lower bounds on the sample complexity of this task and discuss what they mean for the original question of certifying Boson-Sampling.
Section~\ref{sec:upperboundsonthesamplecomplexityinthestatediscriminationsetting} and \ref{sec:samplecomplexityintheblackboxsetting} contain our technical results concerning the sample complexity in the state discrimination and black box setting respectively.
In Section~\ref{sec:thebosonsamplingdistributionisflat} we show that, with very high probability, the Boson-Sampling distribution is exponentially flat (Theorem~\ref{theorem:bosonsamplingdistributionisflat}).
Finally, in Section~\ref{sec:efficientsimulatableinstances} we identify a class of sampling experiments that, due to experimental imperfections, are classically efficiently simulatable in 1-norm.

\section{Setting and notation}\label{sec:setting}
We use the Landau symbols $\landauO$ and $\Omega$ for asymptotic upper and lower bounds.
Moreover, we employ the short hand notation $[j] \coloneqq \{1,\dots,j\}$ for $j \in \Z^+$.
The 1-norm and $\infty$-norm on (probability) vectors are denoted by $\|\cdot\|_1$ and $\|\cdot\|_\infty$.

We consider the output probability distribution $\mcD_U$ of the \emph{Boson-Sampling device} with $n\geq1$ bosons in $m \in \poly(n)$ modes, given by \cite{Aaronson,Scheel}
\begin{equation} \label{eq:bosonsamplingdistribution}
  \Pr_{\mcD_U} \left[S\right] \coloneqq |\bra{1_n} \varphi(U) \ket S |^2 = \frac{|\Perm(U_S)|^2}{\prod_{j=1}^m (s_j!)},
\end{equation}
with $S \in \Phi_{m,n}$ being the output sequence of the Boson-Sampling device where
\begin{equation}
  \Phi_{m,n} \coloneqq \Big\{ (s_1,\dots,s_m) : \sum_{j=1}^m s_j = n \Big\} 
\end{equation}
is the sample space. The state vector $\ket S$ is the Fock space vector corresponding to $S$, $\ket{1_n}$ the initial state vector of the Boson-Sampling device with $1_n \coloneqq (1,\dots, 1,0,\dots, 0)$, $\varphi(U)$ the Fock space (metaplectic) representation of the implemented unitary.
The unitary matrix $U \in U(m)$ is the corresponding unitary in mode space, transforming vectors of bosonic operators.
In turn, $U_S \in \C^{n\times n}$ is the matrix constructed from $U$ by discarding all but the first $n$ columns of $U$ and then, for all $j \in [m]$, taking $s_j$ copies of the $j^{\text{th}}$ row of that matrix (deviating from the notation of Ref.~\cite{Aaronson}).
The permanent $\Perm$ is defined similarly to how the determinant can be defined via the Leibniz formula, but without the alternating sign.

We refer to $\mcD_U$ as the \emph{Boson-Sampling distribution}.
The \emph{Boson-Sampling problem} is: given as input to the algorithm $n$, $m$, and $U$, 
sample exactly or approximately from $\mcD_U$.
We will also consider the \emph{post-selected Boson-Sampling distribution} $\mcD^*_U$ which is obtained from $\mcD_U$ by discarding all output sequences $S$ with more than one boson per mode, i.e., all $S$ which are not in the set of \emph{collision-free} sequences
\begin{equation}
  \Phi^*_{m,n} \coloneqq \Big\{ S \in \Phi_{m,n} : \forall s \in S : s \in \{0,1\} \Big\} .
\end{equation}
For the relevant scalings of $m$ with $n$ the post-selection can be done efficiently in the sense that on average at least a constant fraction of the outcome sequences is collision-free (Theorem~13.4 in Ref.~\cite{Aaronson}).

The main result of Ref.~\cite{Aaronson} is that under reasonable complexity theoretic conjectures, 1-norm approximate Boson-Sampling, i.e., sampling from a distribution that is close to the Boson-Sampling distribution $\mcD_{U}$ in 1-norm, is computationally hard, with high probability if the unitary $U$ is chosen from the Haar measure $\mu_H$, which we denote by $U \sim \mu_H$, and $m$ increases sufficiently fast with $n$.
In fact, the hardness result of approximate Boson-Sampling requires that $m \in \Omega(n^5)$, but it is conjectured that $m$ growing faster than $n^2$ is sufficient.
Importantly, the proof of this result considers only collision-free output sequences, so in fact approximately sampling from $\mcD^*_U$ is already hard and the hardness argument for Boson-Sampling only uses the structure of the distribution $\mcD_U$ on $\Phi^*_{m,n}$.

We will repeatedly use that the size
\begin{equation}
  |\Phi_{m,n}| = \binom{m+n-1}{n}
\end{equation}
of the sample space of Boson-Sampling, which grows faster than exponentially with $n$, fulfills the following bound: 
Let $m \leq c\,n^{\nu}$ for some $\nu\geq1$ and $c\geq0$, then 
\begin{align} \label{eq:samplespaceishuge}
  |\Phi_{m,n}| & \leq \frac{(m+n-1)^n}{n!} \leq \left(\frac{(m+n-1)\,\e}{n}\right)^n \\
  &\leq \e^n\,(c\,n^{\nu-1} + 1 - 1/n )^n \leq  (2\,(c+1)\,\e)^n\,n^{(\nu-1)\,n} . \label{eq:Phibound}
\end{align}

\section{Boson-Sampling in the context of sample complexity}
\label{sec:bosonsamplinginthecontextofsamplecomplexity}
Our aim is to understand if and how an experimental implementation of Boson-Sampling can be certified to sample from the correct distribution.
In particular, we address the question of whether a certification can be achieved from the samples output by the device only.
This is natural because Boson-Sampling is an abstract sampling problem with a classical input (the number of bosons $n$, the number of modes $m$ the unitary matrix $U$) and a classical output (samples from $\Phi_{m,n}$).
The sampling problem as such is independent of the particular physical implementation.
We call all information about a claimed Boson-Sampling device that is in principle available to the certifier in addition to the samples \emph{a priori knowledge}.

The problem of certifying a device can be formalized as a decision problem.
Whether a decision can be reached can be expressed as a statement about the existence or non existence of an algorithm which, possibly using parts or all of the a priori knowledge, and given a set of samples, accepts (device certified) in the ``good'' situation with probability at least $2/3$ and in the bad instances rejects (device not certified) with probability at least $2/3$.
The probabilities to erroneously reject in a ``good'' instance (accept in a ``bad'' instance) are called errors of the first (second) kind.
If such an algorithm exists we say it decides the problem, if no such algorithm exists we say that the problem can not be decided.
If a sampling problem has a natural problem size, like the number of bosons $n$ in the case of Boson-Sampling, it is natural to consider the scaling of the number of samples needed such that a deciding algorithm exists as a function of this problem size.
The order of the number of samples needed by an algorithm is called its \emph{sample complexity}.
The sample complexity of a decision problem in turn is the minimal sample complexity of any algorithm that decides the problem.
The choice of the value $2/3$ for the accept/reject probabilities is purely conventional.
Any other constant finite bias in the accept/reject probabilities can be amplified to values arbitrarily close to one without changing the sample complexity.

The main hardness result of Ref.~\cite{Aaronson} covers all distributions that are 1-norm close to the ideal Boson-Sampling distribution.
Hence, an algorithm that certifies a Boson-Sampling device must necessarily reject with probability at least $2/3$ whenever the device samples from a distribution further away than some small constant distance in 1-norm and it is desirable that it accepts with probability at least $2/3$ if the device samples from the ideal Boson-Sampling distribution.
Such an algorithm must hence at least be able to decide whether a given device samples from the ideal Boson-Sampling distribution or the uniform distribution over the same sample space.

In the state discrimination setting, the sample complexity of this task is of order $\landauO(n^3)$ (Theorem~\ref{theorem:samplecomplexityinthestatediscriminationsetting}), but it is certainly not realistic to assume that the certifier asked to decide this question has full knowledge of the ideal Boson-Sampling distribution.
After all, it is the very point of Boson-Sampling that approximating the probabilities of individual outcomes is a computationally hard problem.
It is therefore important to investigate the sample complexity of this task under more reasonable restrictions on the a priori knowledge and computational power of the certifier.

Colloquially speaking, our results on this problem, which are formally stated below, give rise to a rather ironic situation:
Instead of building a device that implements Boson-Sampling, for example by means of a quantum optical experiment, one could instead simply program a classical computer to efficiently sample from the uniform distribution over $\Phi^*_{m,n}$ and claim that the device samples from the post-selected Boson-Sampling distribution $\mcD^*_U$ for some $U$.
If one chooses $U$ from the Haar measure the chances of being caught cheating becomes significantly large only after one was asked for exponentially many samples.
This implies that the findings of any experimental realisation of Boson-Sampling have to be interpreted with great care, as far as the  notion ``quantum supremacy'' is concerned.

To be precise, our main result is a lower bound on the sample complexity of distinguishing the post selected Boson-Sampling distribution from the uniform one for symmetric probabilistic algorithms.
We will give a precise definition of symmetric probabilistic algorithms in Section~\ref{sec:samplecomplexityintheblackboxsetting} (Definition~\ref{definition:symmetricprobabilisticalgorithm}), but essentially a probabilistic decision algorithm is called symmetric if its output distribution is invariant under relabeling the elements of the sample space.
\begin{theorem}[Distinguishing the post selected Boson-Sampling distribution from the uniform one] \label{theorem:mainresultone}
  If $U \sim \mu_H$, i.e., $U$ is drawn from the Haar measure, and $m \in \Omega(n^\nu)$ with $\nu>2$, then with probability supra-exponentially small in $n$ no symmetric probabilistic algorithm 
  can distinguish the post-selected Boson-Sampling distribution $\mcD^*_U$ from the uniform distribution on $\Phi^*_{m,n}$ from fewer than $\Omega(\e^{n/2})$ many samples.
\end{theorem}
\begin{proof}
  The theorem is an immediate corollary of our Theorems~\ref{theorem:symmetricalgorithmsandflatdistributions} and \ref{theorem:bosonsamplingdistributionisevenflateronthecollisionfreesubspace}.
\end{proof}
Notice that the hardness results of Ref.~\cite{Aaronson} requires that $\nu > 5$, and $\nu>2$ is known to be necessary for the proof strategy used there to work, so our theorem fully covers the interesting parameter range.

Without any a priori knowledge about the distribution the labels of the elements of the sample space have no meaning to the certifier. 
Thus, in the black box setting any decision reached following a non-symmetric algorithm seems arbitrary and cannot qualify as a conclusion, but at the same time, as Theorem~\ref{theorem:mainresultone} shows, symmetric algorithms are essentially useless to distinguish the post selected Boson-Sampling distribution from the uniform one.

We now argue that the above theorem is relevant for the problem of certifying a real Boson-sampling device.
Importantly, in a realistic situation the certifier knows the specific unitary $U$ implemented by the supposed Boson-Sampling device.
In some particular cases, for example, when $U$ is such that it has some special structure, e.g., such that some outcomes are particularly likely to occur or some have probability zero \cite{tichy2010,tichy2012}, this knowledge could be used to construct a non-symmetric decision algorithm, thus opening up the possibility to drastically reduce the sample complexity.
However, this seems very implausible in the interesting instances, i.e., the ones that are covered by the hardness proof, precisely due to the fact that it is believed to be $\#P$ hard to approximate the probabilities of individual outcomes for $U \sim \mu_H$.

We can make a similar statement about the full Boson-Sampling distribution:
\begin{theorem}[Distinguishing the Boson-Sampling distribution from the uniform one] \label{theorem:mainresulttwo}
  If $U \sim \mu_H$ and $m \in \Omega(n^\nu)$ with $\nu>3$, then with probability supra-exponentially small in $n$ no symmetric probabilistic algorithm 
  can distinguish the Boson-Sampling distribution $\mcD_U$ from the uniform distribution over $\Phi_{m,n}$ from fewer than $\Omega(\e^{n/2})$ many samples. 
\end{theorem}
\begin{proof}
  The theorem is an immediate corollary of Theorem~\ref{theorem:symmetricalgorithmsandflatdistributions} and \ref{theorem:bosonsamplingdistributionisflat}.
\end{proof}
As said earlier, Ref.~\cite{Aaronson} requires that $\nu > 5$, however, it is believed that $m$ growing faster than quadratically with $n$ is sufficient for hardness, which leaves open a parameter range not covered by Theorem~\ref{theorem:mainresulttwo}.
At the same time we have good reasons to believe that this is merely a technicality and that $\nu>2$ is already sufficient for the statements of Theorem~\ref{theorem:mainresulttwo} to be valid (see the discussion in Section~\ref{sec:thebosonsamplingdistributionisflat} and Theorem~\ref{theorem:bosonsamplingdistributionisevenflateronthecollisionfreesubspace}).

In the latter case of the full Boson-Sampling distribution $\mcD_U$ the restriction to symmetric algorithms is arguably less natural, mainly because it is known that bosons tend to \emph{bunch} or \emph{cluster} \cite{hong87,tichy2010,tichy2012}.
That is, output sequences $(s_1,\dots,s_m) \in \Phi_{m,n}$ with ``collisions'', i.e., ones in which at least one $s_j$ is larger than one are, on average over $U \sim \mu_H$, more likely than in the uniform distribution over $\Phi_{m,n}$ (although not dramatically more likely, see Theorem~13.4 in Ref.~\cite{Aaronson}).
This could potentially be used to distinguish the Boson-sampling distribution from the uniform one by a non-symmetric algorithm.
We argue that this does not qualify as a certification of a provably hard task.
This is because the proof of Ref.~\cite{Aaronson} only considers the distribution on the Boson-Sampling distribution on the  collision-free sector $\Phi^*_{m,n}$.
Hence, checking that the output distribution shows the correct bunching cannot help to corroborate that the output distribution is covered by the hardness proof of Ref.~\cite{Aaronson}.

This is related to another subtlety that is important to correctly understand the meaning of our results.
The hardness result of Ref.~\cite{Aaronson} covers all distributions that are at most a small distance away in 1-norm from the ideal Boson-Sampling distribution.
Hence, a device that certifies that a black box samples from a distribution that is covered by the hardness results of Ref.~\cite{Aaronson} does not necessarily need to accept with probability at least $2/3$ on the ideal Boson-Sampling distribution, it is in principle sufficient if it does so on some distribution inside this 1-norm ball.
The 1-norm ball includes distributions that are not exponentially flat and which can be distinguished from the uniform distribution from polynomially large number of samples using a symmetric algorithm \cite{Batu}.
Symmetric certification algorithms with polynomial sample complexity for these distributions thus cannot be excluded by our results.

There is a further subtlety:
Consider a device that with probability $1-\epsilon$ outputs a sample from an ideal Boson-Sampling device and with probability $\epsilon$ outputs a specific sample that encodes the solution to an NP-complete problem.
The output distribution of this device would be $\epsilon$ close to the Boson-Sampling distribution in 1-norm.
At the same time, it can be certified from $\landauO(1/\epsilon)$ many samples, using a simple but non-symmetric algorithm, that the device is implementing a hard sampling problem by simply identifying the special outcome and checking that it is indeed a solution to the NP-complete problem.
Even though this distribution is covered by the hardness results of Ref.~\cite{Aaronson}, one would hardly say that its (certifiable) hardness is a consequence of the hardness of Boson-Sampling.

One might also consider the following alternative certification scenario.
Assume one already has a certified Boson-Sampling device, then one could try to certify another device by comparing the samples they output.
Again, our technical results, Theorem~\ref{theorem:symmetricalgorithmsandflatdistributions} and \ref{theorem:bosonsamplingdistributionisflat}, imply that with high probability this cannot be done using a symmetric algorithm and less than exponentially many samples.

Finally, it is important to note that our findings do not contradict the results of Ref.~\cite{Aaronson}.

\section{Upper bounds on the sample complexity in the state discrimination setting}
\label{sec:upperboundsonthesamplecomplexityinthestatediscriminationsetting}
In the state discrimination setting the certifier has the promise that the given sampling device samples from one of two known distributions $\mcP$ or $\mcQ$.
In particular he has knowledge of the sample space and all the probabilities that each of the two candidate distributions assign to the elements of this space.
The certifier's aim is to minimize the probability of wrongfully answering $\mcQ$ if the true distribution is $\mcP$ (error of the first kind) and that of wrongfully answering $\mcP$ if the true distribution is $\mcQ$ (error of the second kind).
This minimization can be done in various different ways.
For example one can minimize the (weighted) sum of the two probabilities, or minimize one while the other is kept constant or suppressed exponentially in the number of samples with a predefined rate.
The asymptotic behavior of the number of samples needed in these situations has been extensively studied in both the classical \cite{Chernoff52,Blahut1974} and quantum setting \cite{Hellstroem76,Holevo78,Hiai91,Audenaert2007}, in which the probability distributions are replaced by quantum states.
See also the introduction of Ref.~\cite{Audenaert2012} for a short review and \cite{Cover91} for further references.

Only recently, in Ref.~\cite{Audenaert2012}, bounds on the error probabilities for finite sample sizes were derived.
They hold in both the classical and the quantum setting, but here we will only need the classical versions.
They imply that in the state discrimination setting the Boson-Sampling distribution can be distinguished from the uniform distribution from a polynomial number of samples:
\begin{theorem}[Lower bound on the sample complexity of distinguishing the Boson-Sampling distribution from the uniform one in the state discrimination setting] \label{theorem:samplecomplexityinthestatediscriminationsetting}
  Let $\epsilon >0$ and $m \leq c\,n^\nu$ for some $c\geq0,\nu\geq1$.
  Then for any $\gamma>0$ there exists a constant $C>0$, such that for and any instance of Boson-Sampling with $n$ bosons in $m$ modes, whose distribution is at least $\epsilon$ far from the uniform distribution in 1-norm, there is an algorithm that distinguishes the former from the latter from $C\,n^{2+\gamma}$ many samples.
\end{theorem}
Note that the above theorem covers all instances of Boson-Sampling that can potentially be hard to sample from approximately in 1-norm.
\begin{proof}
  For $0 \leq t \neq 1$ and two probability distributions over a finite sample space $\Phi$ we define the $t$-\emph{R\'{e}nyi relative entropy} of $\mcP$ given $\mcQ$
  \begin{equation}
    \Sr[t] {\mcP}{\mcQ} \coloneqq 
    \begin{dcases}
    \frac{1}{t-1} \ln \quad\sum_{\mathclap{S \in \supp(\mcP) \cap \supp(\mcQ)}}\quad \Pr_{\mcP}[S]^t \Pr_{\mcQ}[S]^{1-t} & \text{if } \supp(\mcP) \subseteq \supp(\mcQ) \\
    \infty & \text{otherwise} 
  \end{dcases},
\end{equation}
where $\supp(\mcP) = \{S \in \Phi : \Pr_{\mcP}[S] \neq 0 \}$ and we set $\ln 0 = -\infty$.
As $\mcP$ is normalized the limit $t \to 1$ exists \cite{Audenaert2012} and $\Sr{\mcP}{\mcQ} \coloneqq \lim_{t \to 1} \Sr[t]{\mcP}{\mcQ}$ is called \emph{relative entropy} of $\mcP$ given $\mcQ$.

Let $\beta_{l,\alpha}$ be the optimal achievable error of the second kind in the state discrimination setting after receiving $l$ samples when the error of the first kind is upper bounded by $\alpha$.
Theorem~3.3 in Ref.~\cite{Audenaert2012} implies that for all $l,\alpha>0$ 
\begin{equation}
  \frac{1}{l} \ln \beta_{l,\alpha} \leq - \Sr {\mcP} {\mcQ} + \frac{1}{\sqrt{l}} 4 \sqrt{2} \ln(\alpha^{-1}) \ln \eta - \frac{2 \ln 2}{l} ,
\end{equation}
where
\begin{equation}
  \eta = 1 + \e^{\Sr[3/2]{\mcP}{\mcQ}/2} + \e^{-\Sr[1/2]{\mcP}{\mcQ}/2} .
\end{equation}
Since
\begin{align}
  \eta &\leq 2 + \e^{\Sr[3/2]{\mcP}{\mcQ}/2} \leq \e^{\Sr[3/2]{\mcP}{\mcQ}/2+\ln 3}
\end{align}
this implies
\begin{align} \label{eq:simplefinitesampleerrorbound}
  \frac{1}{l} \ln \beta_{l,\alpha} \leq - \Sr {\mcP} {\mcQ} + \frac{1}{\sqrt{l}} 4 \sqrt{2} \ln(\alpha^{-1})\, (\Sr[3/2]{\mcP}{\mcQ}/2+\ln 3) .
\end{align}
Theorem~1.15 in Ref.~\cite{Ohya93} implies the first of the following inequalities
\begin{equation}
  \Sr{\mcP}{\mcQ} \geq \frac{1}{2} \|\mcP - \mcQ \|_1^2 \geq \frac{1}{2} \epsilon^2 ,
\end{equation}
the second is implied by the assumptions of the Theorem.
Moreover, if $\mcQ$ is the uniform distribution over $\Phi$, then
\begin{align}
  \Sr[3/2]{\mcP}{\mcQ} &\leq \Sr[2]{\mcP}{\mcQ} = \ln(|\Phi| \sum_{S \in \Phi} \Pr_{\mcP}[S]^2) \leq \ln |\Phi| .
\end{align}
Hence, for $\Phi = \Phi_{m,n}$ and if $m \leq c\,n^\nu$ we have by Eq.~\eqref{eq:Phibound}
\begin{equation}
  \Sr[3/2]{\mcP}{\mcQ} \leq n \ln(2\,(c+1)\,\e) + n\,(\nu-1) \ln(n) .
\end{equation}
This implies that for $\alpha = 1/3$, $\mcP$ the Boson-Sampling distribution with $m \leq c\,n^{\nu}$, and $\mcQ$ the uniform distribution over $\Phi_{m,n}$, a number of samples $l$ scaling like $l \in \Omega(n^{2+\gamma})$, for any $\gamma>0$, is sufficient to make the right hand side of Eq.~\eqref{eq:simplefinitesampleerrorbound} negative, and thereby $\beta_{l,1/3} \leq 1/3$ for sufficiently large $n$.
\end{proof}

\section{Sample complexity in the black box setting}
\label{sec:samplecomplexityintheblackboxsetting}
In this section we give lower bounds on the sample complexity of decision problems in the black box setting.
Apart from the scenario relevant for the certification of Boson-Sampling, in which the certifier is given a black box and is asked to distinguish the two cases where it samples from the Boson-Sampling distribution or the uniform one, we will also cover scenarios where the certifier is given two or more black boxes and is asked to decide whether they sample from the same or form different probability distributions (see also Ref.~\cite{Batu}).

The lower bounds are ultimately a consequence of a variant of the birthday paradox for $\epsilon$-flat probability distributions.
We call a probability distribution $\mcP$ over a finite sample space \emph{$\epsilon$-flat} if $\|\mcP\|_\infty \leq \epsilon$, i.e., all probabilities are smaller than $\epsilon$, or equivalently if $\mcP$ has min entropy $H_\infty \geq - \log_2 \epsilon$.

\begin{lemma}[Non-uniform non-identically distributed birthday paradox] \label{lemma:nonuniformnoniidbirthdayparadox}
  The probability $\bar p(l,|\Phi|,\epsilon)$ that $l$ samples drawn independently from not necessarily identical $\epsilon$-flat distributions over a finite sample space $\Phi$ are all different fulfills
  \begin{equation} \label{eq:nonuniformbirthdayparadoxlowerbound}
    \forall l \leq 1 + 1/(2\epsilon):\quad \bar p(l,|\Phi|,\epsilon) \geq 2^{-l^2 \epsilon} .
  \end{equation}
\end{lemma}
\begin{proof}
  The probability $\bar p(l,|\Phi|,\epsilon)$ that all samples are different is bounded by
  \begin{align}
    \bar p(l,|\Phi|,\epsilon) &\geq \prod_{j=1}^{l-1} (1 - j \epsilon ) .
  \end{align}
  If $1/2 \leq 1 - \epsilon\,(l-1) \leq 1$, we have
  \begin{equation}
    \forall j\in[l-1]:\quad  1-j\,\epsilon \geq 2^{-2\,j\,\epsilon} .
  \end{equation}
  This implies that for sufficiently large $l$
  \begin{align}
    \bar p(l,|\Phi|,\epsilon) &\geq \prod_{j=1}^{l-1} (1 - j\,\epsilon ) \geq \prod_{j=1}^{l-1} 2^{-2\,j\,\epsilon} \\
    &= 2^{-2 \sum_{j=1}^{l-1} j\,\epsilon} = 2^{-l\,(l-1) \epsilon} \geq 2^{-l^2 \epsilon} .
  \end{align}
\end{proof}
Now, we consider the situation of $\k$ black boxes that sample each from one of $\k$ probability distributions $(\mcP^{(j)})_{j=1}^\k$ over the same finite sample space $\Phi$.
For $j\in [\k]$ and $l \in \Z^+$, let $\mcS^{(j)} \coloneqq (S^{(j)}_1,\dots,S^{(j)}_l) \in \Phi^l$ be sequences of samples of length $l$ from each of the distributions respectively.
We will keep the discussion in this chapter general but will later mostly be interested in the case $\k = 1$.

The certifier works under the assumption that the sampling device outputs independent identically distributed samples.
Hence, the order of the samples in each sequence should not influence the certifier's decision.
Moreover, in the black box setting the certifier is assumed to have no a priori knowledge about the distribution.
If in addition the decision problem of the certifier is invariant under a relabeling of the sample space, its decision should be independent of which element of the sample space is assigned which label.
If this is not the case it cannot qualify as a conclusion reached based on the samples.
Therefore, for tasks such as deciding whether a given black box is sampling from the uniform distribution or not, or deciding whether a number of black boxes sample from the same or from different distributions the certifier should follow a \emph{symmetric probabilistic algorithm}.
\begin{definition}[Symmetric probabilistic algorithm] \label{definition:symmetricprobabilisticalgorithm}
  An algorithm that takes as input for each $j \in [\k]$ a sequence of samples $\mcS^{(j)} \subset \Phi^l$ and probabilistically outputs either ``accept'' or ``reject'' is called a \emph{symmetric probabilistic algorithm} if its output distribution is invariant under permuting the samples in each sequences
  \begin{align}
    \forall j &\in [\k]:&  (S^{(j)}_1,\dots,S^{(j)}_l) &\mapsto (S^{(j)}_{\tau_j(1)},\dots,S^{(j)}_{\tau_j(l)}) ,\quad \tau_j \in \Sym([l]) ,
  \intertext{and relabeling of the sample space $\Phi$, i.e., the action of $\Sym(\Phi)$ on all $\mcS^{(j)}$ simultaneously}
  \forall j &\in [\k]:&  (S^{(j)}_1,\dots,S^{(j)}_l) &\mapsto (\kappa(S^{(j)}_1),\dots,\kappa(S^{(j)}_l)) ,\quad \kappa \in \Sym(\Phi) .
  \end{align}
\end{definition}

Following Ref.~\cite{Batu} we define the \emph{fingerprint tensor} $C((\mcS^{(j)})_{j=1}^\k) \in \N^{(l+1)\times\dots\times(l+1)}$ of the sequences of samples, such that for all $k_1,\dots,k_\k \in \{0,1,\hdots,l\}$, $C_{k_1,\dots,k_\k}$ is the number of elements in $\Phi$ that for all $j \in [\k]$ appear exactly $k_j$ times in the $j$-th sequence of samples $\mcS^{(j)}$.
Obviously $\sum_{k_1,\dots,k_\k=0}^l C_{k_1,\dots,k_\k} = |\Phi|$.
For $\k=1$ this construction results in the fingerprint vector
\begin{equation}
  C_{k_1} \coloneqq |\{ S' \in \Phi : |\{S \in \mcS^{(1)} : S = S' \}| = k_1 \}| 
\end{equation}
and for $\k=2$ the result is the fingerprint matrix
\begin{equation}
  \begin{split}
    C_{k_1,k_2} \coloneqq |\{ S' \in \Phi : &|\{S^{(1)} \in \mcS^{(1)} : S^{(1)} = S' \}| = k_1 \\
    \text{and } &|\{S^{(2)} \in \mcS^{(2)} : S^{(2)} = S' \}| = k_2 \}| .
  \end{split}
\end{equation}
For example, if $\Phi=[6]$,
$S^{(1)}= (1,5,1,1,2)$ and $S^{(2)} = (2,6,1,4,6)$, then the fingerprint matrix is given by
\begin{equation}
	C= \left(
	\begin{array}{cccccc}
	1 & 1 & 1 & 0 & 0 & 0\\
	1 & 1 & 0 & 0 & 0 & 0\\
	0 & 0 & 0 & 0 & 0 & 0\\
	0 & 1 & 0 & 0 & 0 & 0\\
	0 & 0 & 0 & 0 & 0 & 0\\
	0 & 0 & 0 & 0 & 0 & 0\\
	\end{array}
	\right).
\end{equation}

The fingerprint tensor encodes all the information contained in the samples that is invariant under permuting the labels of the sample space or reordering the samples in each sequence. That is, the sequences of samples can be reconstructed exactly from the fingerprint up to a permutation of the samples in each sequence and a global relabeling of the sample space \cite{Batu}.
This proves:
\begin{lemma}[Symmetric algorithms and the fingerprint (see also Ref.~\cite{Batu})] \label{lemma:symmetricalgorithmsandthefingerprint}
  For every symmetric probabilistic algorithm $\mcA$ there is exists an algorithm $\mcA'$ which has the same output distribution as $\mcA$, but takes as input the fingerprint of the sequences of samples.
\end{lemma}

Denote by $\mcD_C((\mcP^{(j)})_{j=1}^\k, l)$ the probability distribution on fingerprint tensors induced by drawing $l$ samples from each $\mcP^{(j)}$, and then constructing the corresponding fingerprint tensor, and when we write $C \sim \mcD_C((\mcP^{(j)})_{j=1}^\k, l)$ we mean $C$ drawn from $\mcD_C((\mcP^{(j)})_{j=1}^\k, l)$.
For each $|\Phi|$, $\k$, and $l$ there is a unique \emph{trivial fingerprint} tensor $\tilde C$ that characterizes the situation where no sample appears more than once.
For $\k=1$ this is the vector
\begin{equation}
  \tilde C = (|\Phi|-l,l,0,\hdots,0) \in \N^{l+1} ,
\end{equation}
and for $\k=2$ it is the matrix
\begin{equation} \label{eq:standardfingerprintmatrix}
  \tilde C \coloneqq
  \begin{pmatrix}
    |\Phi| - 2\,l & l & 0 & \cdots &0\\
    l & 0 & 0 & \cdots &0\\
    0 & 0 & 0 &\cdots &0 \\
    \vdots& \vdots &  \vdots& \ddots& \vdots\\ 
      0 & 0 & 0 & \dots &  0 \\
  \end{pmatrix} \in \N^{(l+1)\times(l+1)} .
\end{equation}

Due to the birthday paradox, fingerprint tensors constructed from few samples are trivial with high probability.
\begin{lemma}[Fingerprint tensors from few samples] \label{lemma:kfingerprintmatricesfromfewsamples}
  Let $\k\in \Z^+$ and $(P^{(j)})_{j=1}^\k$ be $\epsilon$-flat probability distributions over a finite sample space $\Phi$.
  If $l \in \landauO((1/\epsilon)^{1/4})$ many samples are drawn from each $\mcP^{(j)}$ then
  \begin{equation} \label{lemmaFingerprint}
    \Pr_{C \sim \mcD_C((\mcP^{(j)})_{j=1}^\k,l)} [C \neq \tilde C] \in \landauO(\k^2 \sqrt{\epsilon}) .
  \end{equation}
\end{lemma}
\begin{proof}
  Let $a>0$ and $l \leq a\,(1/\epsilon)^{1/4}$ and denote for each $j \in [\k]$ by $\mcS^{(j)}$ the sequence of $l$ samples drawn from $\mcP^{(j)}$.
  Since all the $\mcP^{(j)}$ are $\epsilon$-flat probability distributions over $\Phi$ we have
  \begin{equation} \label{eq:boundoncolisionporbability}
    \Pr_{C \sim \mcD_C((\mcP)_{j=1}^\k, l)} \left[\exists S' \in \Phi : \sum_{j=1}^\k |\{ S \in \mcS^{(j)} : S = S'\}| > 1  \right] = \bar p(\k\,l,|\Phi|,\epsilon) ,
  \end{equation}
  with $\bar p$ as in Lemma~\ref{lemma:nonuniformnoniidbirthdayparadox}.
  For sufficiently small $\epsilon$ Lemma~\ref{lemma:nonuniformnoniidbirthdayparadox} yields
  \begin{equation}
    \bar p(\k\,l,|\Phi|,\epsilon) \leq 2^{-(\k\,a)^2 \sqrt{\epsilon}} 
  \end{equation}
  and
  \begin{equation}
    \Pr_{C \sim \mcD_C((\mcP)_{j=1}^\k, l)} [C \neq \tilde C] \leq 1 - 2^{-(\k\,a)^2 \sqrt{\epsilon}} \leq (\k\,a)^2 \sqrt{\epsilon}.
  \end{equation}
\end{proof}
Similar results can be obtained for all scalings $l \in \landauO((1/\epsilon)^\alpha)$ with $\alpha < 1/2$, but $\alpha = 1/4$ is good enough for our purposes and yields the particularly simple result stated above.

\begin{theorem}[Symmetric algorithms and $\epsilon$-flat distributions] \label{theorem:symmetricalgorithmsandflatdistributions}
  For every symmetric probabilistic algorithm there exists a trivial output distribution such that the output distribution of the algorithm after receiving at most $\landauO((1/\epsilon)^{1/4})$ many samples from each of $\k$ black boxes sampling from $\epsilon$-flat distributions is with probability $1 - \landauO(\k^2 \sqrt{\epsilon})$ equal to the trivial output distribution and hence, in particular, does not depend on which $\epsilon$-flat distributions were used to generate the samples.
\end{theorem}
\begin{proof}
  By Lemma~\ref{lemma:symmetricalgorithmsandthefingerprint} any symmetric probabilistic algorithm is equivalent to an algorithm that only receives the fingerprint of the samples as input.
  If all input distributions are $\epsilon$-flat, then by Lemma~\ref{lemma:kfingerprintmatricesfromfewsamples}, if at most $\landauO((1/\epsilon)^{1/4})$ samples are drawn from each distribution, the probability that their fingerprint is non-trivial is of order $\landauO(\k^2 \sqrt{\epsilon})$.
  The result follows and the trivial output distribution is the output distribution corresponding to samples  with the trivial fingerprint.
\end{proof}
By strengthening Lemma~\ref{lemma:kfingerprintmatricesfromfewsamples}, as pointed out after its proof, a result similar to Theorem~\ref{theorem:symmetricalgorithmsandflatdistributions} can be obtained for the number of samples scaling like $\landauO((1/\epsilon)^\alpha)$ for all $\alpha < 1/2$.

\section{The Boson-Sampling distribution is flat}
\label{sec:thebosonsamplingdistributionisflat}
In this section we show that the Boson-Sampling distribution is extremely flat with high probability.
The strategy is as follows:
First we relate the probability measure induced on the matrices $U_S$ described in Section~\ref{sec:setting} to a Gaussian measure $\mu_{G_S(\sigma)}$.
Then we use measure concentration for $\mu_{G_S(\sigma)}$ to prove $\epsilon$-flatness.

A crucial step in the proof of the main result of Ref.~\cite{Aaronson} is to show that if $m$ is sufficiently large compared to $n$ and $U \sim \mu_H$, i.e., $U$ is chosen from the Haar measure $\mu_H$ on $U(m)$, then, for any fixed $S \in \Phi^*_{m,n}$, the measure on $\C^{n\times n}$ induced by the map $g_S = (U \mapsto U_S)$ is close to $\mu_{G(1/\sqrt{m})}$, where $\mu_{G(\sigma)}$ is the measure obtained by choosing the real and imaginary part of every entry of an $n\times n$ matrix independently from a Gaussian distribution with mean zero and standard deviation $\sigma$.
\begin{lemma}[Theorem~5.2 in Ref.~\cite{Aaronson}]\label{lemma:multiplicativeerrorbound}
  Let $f:\C^{n\times n} \to \left[0,1\right]$ be measurable and $\delta>0$ with the property that $m \geq (n^5/\delta) \ln^2(n/\delta)$.
  Then
  \begin{equation} \label{eq:multiplicativeerrorbound}
    \forall S \in \Phi^*_{m,n}:\quad \E_{U \sim \mu_H} f(U_S)  \leq (1 + \landauO(\delta)) \E_{X \sim \mu_{G(1/\sqrt{m})}} f(X) .
  \end{equation}
\end{lemma}
It is known that $m \geq c\,n^{\nu}$ with $\nu > 2$ and $0< c \in \landauO(1)$ is necessary for closeness of $\mu_H \circ g_S^{-1}$ and $\mu_{G(1/\sqrt{m})}$.
As this is a crucial ingredient to the proof of hardness of Ref.~\cite{Aaronson}, 
we will from now on assume that $m \geq c\,n^{\nu}$ with $\nu > 2$ and $0 < c \in \landauO(1)$. 

Lemma~\ref{lemma:multiplicativeerrorbound} is not strong enough for our purpose, as we must be able to control all of $\Phi_{m,n}$ and not only the collision-free subspace.
Fortunately, the above lemma extends naturally to all $S \in \Phi_{m,n}$, but first we need some notation:
For every sequence $S$, let $\tilde S$ be the sequence obtained from $S$ by removing all the zeros, i.e,
\begin{equation}\label{eq:tildeSdef}
  \tilde S = (\tilde{s}_1,\dots, \tilde{s}_{|\tilde S|}) \coloneqq (s \in S: s > 0) .
\end{equation}
Let $\mu_{G_S(\sigma)}$ be the probability measure on $\C^{n \times n}$ obtained by drawing the real and imaginary part of every entry of a $|\tilde S| \times n$ matrix independently from a Gaussian distribution with mean zero and standard deviation $\sigma$ and then for all $j \in [|\tilde S |]$ taking $\tilde{s}_j$ copies of the $j^{\text{th}}$ row of this matrix.
\begin{lemma}[Multiplicative error bound]\label{lemma:multiplicativeerrorbound2}
  Let $f:\C^{n\times n} \to \left[0,1\right]$ be measurable and $\delta>0$ with the property that $m \geq (n^5/\delta) \ln^2(n/\delta)$. Then for all $S \in \Phi_{m,n}$
  \begin{equation} \label{eq:multiplicativeerrorbound2}
    \E_{U \sim \mu_H}f(U_{S}) \leq (1 + \landauO(\delta)) \E_{X \sim \mu_{G_S(1/\sqrt{m})}} f(X).
  \end{equation}
\end{lemma}
\begin{proof}
  Let $S \in \Phi_{m,n}$, $\tilde S$ as in Eq.~\eqref{eq:tildeSdef} and $m'\coloneqq |\tilde S|$.
  Define $v$ to be the sequence containing $\tilde s_j$ times the integer $j$ for every $j \in [m']$ in increasing order and $w$ the sequence containing the positions of each of the first of the repeated rows in $U_S$, i.e.,
  \begin{align}
    v &\coloneqq (\underbrace{1,\ldots,1}_{\tilde s_1},
    \underbrace{2,\ldots,2}_{\tilde s_2},\ldots,\underbrace{m',\ldots, m'}_{\tilde s_m'})
    \in (\Z^+)^n,
    \\
    w &\coloneqq (1,1+\tilde s_1, 1+\tilde s_1+\tilde s_2, \ldots, 1+\sum_{j=1}^{m'-1} \tilde s_j)
    \in (\Z^+)^{m'}.
  \end{align}
  The sequence $v$ defines a linear embedding $\eta : \C^{m' \times n} \to  \C^{n \times n}$ component wise by
  \begin{align}
    \eta(Y)_{i,j} \coloneqq Y_{v_i, j} \quad \forall i, j \in [n],
  \end{align}
  i.e., $\eta(Y)$ has $s_j$ copies of the $j$-th row of $Y$. 
  The sequence $w$ defines a linear projection $\pi : \C^{n \times n} \to \C^{m' \times n}$ by
  \begin{align}
    \pi(X)_{i,j} \coloneqq X_{w_i, j} \quad \forall i \in [m'],\ j \in [n],
  \end{align}
  in particular, $\pi(U_S)$ contains only the first out of each series of the repeated rows in $U_S$. 
  Note that $\eta \circ \pi : \C^{n \times n} \to \C^{n \times n}$ 
  is a projection onto the subspace of matrices that have the same repetition structure as $U_S$.
  Let
  \begin{align}\label{eq:deffS}
    f_S \coloneqq f \circ \eta \circ \pi ,
  \end{align}
  then $f_S(U_S) = f(U_S)$ only depends on the first of the repeated rows in $U_S$ and is independent of all the other rows.
  Since the Haar measure is permutation-invariant,
  \begin{align}
    \E_{U \sim \mu_H} f_S(U_S) &= \E_{U \sim \mu_H} f_S(U_{1_n}).
  \end{align}
  Hence, Lemma~\ref{lemma:multiplicativeerrorbound} yields the inequality in the calculation
  \begin{align}
    \E_{U \sim \mu_H} f(U_S) 
    &= \E_{U \sim \mu_H} f_S(U_{1_n})
    \\
    &\leq (1 + \landauO(\delta)) \E_{X \sim \mu_{G(1/\sqrt m)}}f_S(X)
    \\
    &=(1 + \landauO(\delta)) \E_{X \sim \mu_{G_S(1/\sqrt m)}}f(X) ,
  \end{align}
  which finishes the proof.
  
\end{proof}
In addition to the multiplicative error bound we need a concentration result for the Gaussian measure $\mu_{G_S(\sigma)}$.
\begin{lemma}[Concentration of the Gaussian measure $\mu_{G_S(\sigma)}$] \label{lemma:concentrationofthegaussianmeasure}
  For all $n,m \in \Z^+$, all $S \in \Phi_{m,n}$ and all $\xi > 0$
  \begin{equation}
    \Pr_{X\sim \mu_{G_S(\sigma)}}\left[ \max_{j,k\in\left[n\right]} |x_{j,k}| \geq \xi \right]
    \leq 1 - \left(1 - \Erfc\left(\frac{\xi}{\sqrt{2}\,\sigma }\right)\right)^{n^2} .
  \end{equation}
\end{lemma}
\begin{proof}
  For Gaussian random variables we have 
  \begin{equation}
    \forall \xi>0,\ j,k\in\left[n\right]:\quad \Pr_{X\sim \mu_{G(\sigma)}} \left[|x_{j,k}| \geq \xi \right] = \Erfc\left(\frac{\xi}{\sqrt{2}\,\sigma}\right)
  \end{equation}
  where 
  \begin{equation}
    \Erfc\left(\frac{\xi}{\sqrt{2}\,\sigma}\right) \coloneqq 2 \int_{\xi}^\infty \frac{\e^{-\frac{x^2}{2\,\sigma^2}}}{\sqrt{2\,\pi\,\sigma^2}} \dd x 
  \end{equation}
  is the complementary error function.
  This implies that
  \begin{equation}
    \forall \xi>0:\quad \Pr_{X\sim \mu_{G(\sigma)}}\left[ \forall j,k\in\left[n\right]: |x_{j,k}| \leq \xi \right] = \left(1- \Erfc\left(\frac{\xi}{\sqrt{2}\,\sigma}\right)\right)^{n^2} .
  \end{equation}
  It is also true that 
  \begin{equation}
    \begin{split}
      \forall S \in \Phi_{m,n}, \xi>0:\quad 
      \Pr_{X\sim \mu_{G_S(\sigma)}}&\left[ \forall j,k\in\left[n\right]: |x_{j,k}|\leq  \xi \right]  \\ \geq 
      \Pr_{X\sim \mu_{G(\sigma)}}&\left[ \forall j,k\in\left[n\right]: |x_{j,k}| \leq \xi \right] ,
    \end{split}
  \end{equation}
  because the additional dependency of the entries of $X \sim \mu_{G_S(\sigma)}$ only decreases the chance of having an exceptionally large entry.
\end{proof}

\begin{theorem}[Flatness of the Boson-Sampling distribution] \label{theorem:bosonsamplingdistributionisflat}
  Let $\nu > 3$. Then for every $m \in \Omega(n^{\nu})$ 
  \begin{equation}
    - \ln \left(\Pr_{U \sim \mu_H}\left[ \exists S\in\Phi_{m,n}: \Pr_{\mcD_{U}}\left[S\right] \geq \e^{-2\,n} \right] \right) \in \landauO\left(n^{\nu-2-1/n}\right) .
  \end{equation}
\end{theorem}
\begin{proof} 
  Using the union bound (also known as Boole's inequality) we obtain that for every $\epsilon > 0$
  \begin{align}
    &\Pr_{U \sim \mu_H}\left[ \exists S\in\Phi_{m,n}: \Pr_{\mcD_{U}}\left[S\right] \geq \epsilon \right] \\
    &\leq \sum_{S\in\Phi_{m,n}} \Pr_{U \sim \mu_H}\left[ \Pr_{\mcD_{U}}\left[S\right] \geq \epsilon \right] \\
    &\leq |\Phi_{m,n}| \max_{S\in\Phi_{m,n}} \Pr_{U \sim \mu_H}\left[ \Pr_{\mcD_{U}}\left[S\right] \geq \epsilon \right] \\
    &= |\Phi_{m,n}| \max_{S\in\Phi_{m,n}} \Pr_{U \sim \mu_H}\left[ \frac{|\Perm(U_{S})|^2}{\prod_{j=1}^m (s_j!)} \geq \epsilon \right] .
  \end{align}
  Applying Lemma~\ref{lemma:multiplicativeerrorbound2} for the $S$ that yields the maximum with $\delta = n$ and the indicator function
  \begin{equation}
    f(U_S) =
    \begin{cases}
      1 & \text{if } \frac{|\Perm(U_{S})|^2}{\prod_{j=1}^m (s_j!)} \geq \epsilon \\
      0 & \text{otherwise} 
    \end{cases}
  \end{equation}
  yields
  \begin{equation}
    \begin{split} 
      &\Pr_{U \sim \mu_H}\left[ \exists S\in\Phi_{m,n}: \Pr_{\mcD_{U}}\left[S\right] \geq \epsilon \right] \\
      &\leq (1+ \landauO(n))\, |\Phi_{m,n}| \max_{S\in\Phi_{m,n}} \Pr_{X \sim \mu_{G_S(1/\sqrt{m})}} \left[ \frac{|\Perm(X)|^2}{\prod_{j=1}^m (s_j!)} \geq \epsilon \right] . \label{eq:lastinequalitybeforecudeboundonpermx}
    \end{split}
  \end{equation}
  Recall that the permanent $\Perm(X)$ of a matrix $X = (x_{j,k}) \in \C^{n \times n }$ is defined as
  \begin{equation} \label{eq:definitionpermanent}
    \Perm(X) \coloneqq \sum_{\tau \in \Sym([n])} \prod_{j=1}^n x_{j,\tau(j)} ,
  \end{equation}
  where $\Sym([n])$ is the symmetric group acting on $[n]$.
  This implies that
  \begin{equation} \label{eq:stupidpermamentbound}
    \frac{|\Perm(X)|^2}{\prod_{j=1}^m (s_j!)} \leq |\Perm(X)|^2 \leq (n!)^2\,\left(\max_{j,k \in \left[n\right]} |x_{j,k}|\right)^{2n} .
  \end{equation}
  Hence, for every $S\in\Phi_{m,n}$ and every $\epsilon > 0$
  \begin{equation}
    \Pr_{X \sim \mu_{G_S(1/\sqrt{m})}} \left[ \frac{|\Perm(X)|^2}{\prod_{j=1}^m (s_j!)} \geq \epsilon \right] 
    \leq \Pr_{X \sim \mu_{G_S(1/\sqrt{m})}} \left[ \max_{j,k \in \left[n\right]} |x_{j,k}| \geq \left(\frac{\sqrt{\epsilon}}{n!}\right)^{1/n}   \right] .
  \end{equation}
  Now we use Lemma~\ref{lemma:concentrationofthegaussianmeasure} with
  \begin{equation}
    \xi = \left(\frac{\sqrt{\epsilon}}{n!}\right)^{1/n},
  \end{equation}
  and Eq.~\eqref{eq:Phibound} to arrive at
  \begin{equation}
    \begin{split}
      &\Pr_{U \sim \mu_H}\left[ \exists S\in\Phi_{m,n}: \Pr_{\mcD_{U}}\left[S\right] \geq \epsilon \right] \\
      &\leq (1+ \landauO(n))\,(2\,(c+1)\,\e)^n\,n^{(\nu-1)n} \left( 1 - \left(1 - \Erfc\sqrt{\frac{c}{2} \frac{\epsilon^{1/n}\,n^{\nu} }{(n!)^{2/n} }}\right)^{n^2} \right). \label{eq:boundonlargeprobabilitystillwitherfc}
    \end{split}
  \end{equation}
  Bounding the complementary error function by \cite{Ermolova}
    \begin{equation}
      \Erfc\left(x\right) \leq \e^{-x^2} ,
  \end{equation}
  we obtain
  \begin{align}
    1 - \left(1 - \Erfc(x) \right)^{n^2} 
    &\leq 1 - \left(1 - \e^{-x^2} \right)^{n^2} 
    = 1 - \sum_{k=0}^{n^2} \binom{n^2}{k}\,(-\e^{-x^2})^k \\
    &= \sum_{k=1}^{n^2} \binom{n^2}{k}\,\e^{-x^2 k}\,(-1)^{k-1} 
    \leq \sum_{k=1}^{n^2} (n^2 \e/k)^k\,\e^{-x^2 k}  \\
    &=\sum_{k=1}^{n^2} (n^2\,\e^{-x^2 + 1} )^k .
  \end{align}
  If $x$ is sufficiently large such that 
  \begin{equation}
    n^2\,\e^{-x^2+1} \leq \frac{1}{2} < 1,
  \end{equation}
  the geometric series converges and
  \begin{align}    
    \sum_{k=1}^{n^2} (n^2 \e^{-x^2+1} )^k &\leq \frac{n^2 \e^{-x^2+1}}{1 - n^2 \e^{-x^2+1} }\\
    &\leq 2\,n^2\,\e^{-x^2+1}.
  \end{align}
  Hence, for the bound \eqref{eq:boundonlargeprobabilitystillwitherfc} to become meaningful it is sufficient that the argument of the square root in the error function grows slightly faster than linear with $n$.
  Because of the bound $n! \leq \e^{1-n}\,n^{n+1/2}$ (a variant of Stirling's approximation) we have for the argument of the square root in Eq.~\eqref{eq:boundonlargeprobabilitystillwitherfc}
  \begin{align}\label{eq:boundfortheroot}
    \frac{c}{2} \frac{\epsilon^{1/n} n^{\nu} }{(n!)^{2/n} } \geq \frac{c}{2} \frac{\epsilon^{1/n} n^{\nu} }{\e^{2/n-2} n^{2+1/n} } = \frac{c}{2} \frac{\epsilon^{1/n}}{\e^{2/n-2}} n^{\nu-2-1/n} ,
  \end{align}
  and with the convenient choice $\epsilon = \e^{-2n}$ it follows that for all $\nu>3$
  \begin{equation}
    \begin{split}
      &\Pr_{U \sim \mu_H}\left[ \exists S\in\Phi_{m,n}: \Pr_{\mcD_{U}}\left[S\right] \geq \e^{-2n} \right] \\
      &\in \landauO\left(n^3\,(2\,(c+1)\,\e)^n\,n^{(\nu-1)n} \exp(-c\,\e^{-2/n}\,n^{\nu-2-1/n} /2) \right) .
    \end{split}
  \end{equation}
\end{proof}
The above proof of Theorem~\ref{theorem:bosonsamplingdistributionisflat} yields the result only for $\nu>3$.
This is a consequence of the $n!$ prefactor introduced in the extremely crude bound on the permanent used in Eq.~\eqref{eq:stupidpermamentbound}.
In fact, it is known that \cite{Aaronson}
\begin{equation}
  \E_{X \sim \mu_{G(1/\sqrt{m})}}[|\Perm(X)|^2] = 2^n\,n!\,m^{-n} ,
\end{equation}
so it seems likely that, the inequality in Eq.~\eqref{eq:stupidpermamentbound} can be replaced by an inequality that is fulfilled with high probability and has a $\sqrt{n!}$ prefactor instead of the $n!$.

For all $S$ in the collision-free subspace $\Phi_{m,n}^*$ we can show the improved bound:
\begin{theorem}[Flatness of the Boson-Sampling distribution on the collision-free subspace] \label{theorem:bosonsamplingdistributionisevenflateronthecollisionfreesubspace}
  Let $\nu > 1$. Then for every $1>\epsilon>0$ and $m \in \Omega(n^{\nu})$ 
  \begin{equation}
    - \ln \left(\Pr_{U \sim \mu_H}\left[ \exists S\in\Phi^*_{m,n}: \Pr_{\mcD_{U}}\left[S\right] \geq \epsilon \right] \right) \in \landauO\left( (\nu-1)\,n \ln n \right) - 2 \ln (1/\epsilon) ,
  \end{equation}
  and in particular
  \begin{equation}
    - \ln \left(\Pr_{U \sim \mu_H}\left[ \exists S\in\Phi^*_{m,n}: \Pr_{\mcD_{U}}\left[S\right] \geq n^{-n/2} \right] \right) \in \landauO\left( (\nu-2)\,n \ln n \right) .
  \end{equation}
\end{theorem}
\begin{proof}
  It is known that \cite{Aaronson,Aaronson2} 
  \begin{equation}
    \E_{X \sim \mu_{G(1/\sqrt{m})}}[|\Perm(X)|^4] = 2^{2n} (n!)^2\,(n+1)\,m^{-2n} .
  \end{equation}
  Hence, by using Markov's inequality for the positive random variable $|\Perm(X)|^4$ with $m = c n^\nu$ we find that for every $\epsilon>0$
  \begin{equation}
    \Pr_{X \sim \mu_{G(1/\sqrt{m})}}[|\Perm(X)|^2 \geq \epsilon ] \leq 2^{2\,n}\,(n!)^2\,(n+1)\,c^{-2n}\,n^{-2\,\nu\,n} \epsilon^{-2} .
  \end{equation}
  Using again the bound $n! \leq \e^{1-n}\,n^{n+1/2}$ this implies 
  \begin{equation}
    \Pr_{X \sim \mu_{G(1/\sqrt{m})}}[|\Perm(X)|^2 \geq \epsilon ] \leq n (n+1) 2^{2\,n}\,\e^{2-2\,n}\,c^{-2\,n}\,n^{2\,(1-\nu)\,n}\,\epsilon^{-2} .
  \end{equation}
  Hence, by Eq.~\eqref{eq:Phibound}
  \begin{equation}
    \begin{split}
      &|\Phi_{m,n}| \max_{S\in\Phi_{m,n}^*} \Pr_{X \sim \mu_{G_S(1/\sqrt{m})}} \left[ \frac{|\Perm(X)|^2}{\prod_{j=1}^m (s_j!)} \geq \epsilon \right] \\
      &\leq (2\,(c+1)\,\e)^n\,n\,(n+1)\,2^{2\,n}\,\e^{2-2\,n}\,c^{-2\,n}\,n^{(1-\nu)\,n}\,\epsilon^{-2} .
    \end{split}
  \end{equation}
  Inserting this into Eq.~\eqref{eq:lastinequalitybeforecudeboundonpermx} and taking the logarithm yields the first bound, the choice $\epsilon = n^{-n/2}$ the second bound.
\end{proof}

A derivation of a similar bound for all $S \in \Phi_{m,n}$ would prove the statement of Theorem~\ref{theorem:bosonsamplingdistributionisflat} under a weaker condition on $\nu$.
We conjecture that the statement of Theorem~\ref{theorem:bosonsamplingdistributionisflat} is true for all $\nu>2$.

\section{Efficiently simulatable instances in 1-norm} \label{sec:efficientsimulatableinstances}
In this section we finally ask the question in what settings one can expect an efficient classical simulation to be feasible even up to a small error in 1-norm.
After all, any experiment will not realise the precise ideal Boson-Sampling setting, but instead an imperfect approximation thereof.
This may provide room for the efficient classical simulation of the output distribution actually obtained.
Subsequently, we will identify a setting of this kind, which resembles those implementable with present-day linear optical circuits.
It is not claimed that the discussed scenario exactly matches realistic experiments, but it does share many features. 
We will show that efficient classical 1-norm approximate sampling is possible under the following conditions:
\begin{description}
\item[Condition 1:] The input state $\ket{1_n}$ is replaced by a Gaussian product state $\rho$ \cite{Gaussian1,Gaussian2}.
  Sources that produce such states are common in quantum optical implementations.
  In practice, many single photon sources provide approximately coherent states or mixed Gaussian states instead of states for which the probability of having more than a single photon is zero.
  If single photon sources are being generated by heralding \cite{Mosley08} a source of (Gaussian) two-mode squeezed states, the argument presented here is still valid:
After all, the entire statistics, including the heralding events, is then classically simulatable.
\item[Condition 2:] The unitary $\varphi(U)$ with $U\in U(m)$, specifying the optical network, is replaced by a Gaussian completely positive map $T: \mcB(\mcH) \to \mcB(\mcH)$,
a Gaussian channel \cite{GaussianChannel}.
Such operations cover the ideal unitary case $\rho\mapsto \varphi(U)\, \rho\, \varphi(U)^\dagger$ as well as situations involving losses in the linear optical network and aberrations due to mode matching issues.
Gaussian completely positive maps are a very accurate modelling of present linear optical experiments.
\item[Condition 3:] Projection onto Fock states is replaced by measurements described by dichotomic POVMs $\{\Pi_0,\Pi_1\}$ 
  with $\Pi_0 + \Pi_1=\1$ (``bucket detector''), 
  where the Wigner function corresponding to the no click event $\Pi_0$ for some fixed $R>0$ is given by 
  \begin{equation}\label{Wi}
    W_{\Pi_0}(r)=
    \begin{cases}
      1/(2\,\pi)& \text{ if } |r|<R \\
      0 & \text{ otherwise} 
    \end{cases} .
  \end{equation}
  This is an idealised model for imperfect photon detectors used in experiments that distinguishes the presence and the absence of photons, taking into account losses and dark counts.
\end{description}
In the latter aspect the model considered here departs the furthest from actual experiments:
While Condition 1 and 2 are usually satisfied to an extraordinarily large extent in quantum optical experiments, Condition 3 constitutes a rather crude approximation of an imperfect detector such as a realistic avalanche photodiode.
Still, it is noteworthy that these conditions are sufficient to arrive at an efficient classical simulation.
Needless to say, other detector models with positive Wigner functions for the POVM elements work equally well.

For a trace class operator $A$ acting on a system with $m$ modes, its Wigner function $W_A: \R^{2m} \rightarrow \R$ is defined as
\begin{equation}
  W_A(r) \coloneqq \frac{1}{\pi^m} \Tr[w( r )\, \Pi^{\otimes m} w( r )^\dagger A]  .
\end{equation}
Here, $\Pi$ is the single mode parity operator, $\{w(r)\}$ the family of Weyl operators, and $r \in \R^{2m}$ a vector collecting the $2m$ phase space variables.
A state is Gaussian if and only if its Wigner function is Gaussian \cite{Gaussian1,Gaussian2}.
Gaussian channels transform states with a Gaussian Wigner function into states with a Gaussian Wigner function.
The Jamiolkowski isomorphs of such maps are Gaussian states.


Expressing Hilbert-Schmidt scalar products as integrals over Wigner functions, one finds for the dark count rate
\begin{equation}
  \langle 0|\Pi_1|0\rangle =  1- \langle 0|\Pi_0|0\rangle= 1- 2 \int_0^R r\,\e^{-r^2} \dd r = \e^{-R^2} .
\end{equation}
The Wigner function of the coherent state $|1\rangle_c$ that contains $1$ photon on average is given by
\begin{equation}
  W_{|1\rangle_c\langle 1|_c} (r) = \frac{1}{\pi} \e^{-(r_1-1)^2-r_2^2} .
\end{equation}	
With this one finds for the effective detector efficiency
\begin{equation}
  \begin{split}
    &1 - \langle 1|_c\Pi_0|1\rangle_c = \\
    &1 - \frac{1}{2 \sqrt{\pi}} \int_{-R}^R \dd p \e^{-p^2} \left( \Erf(1+(R^2-p^2)^{1/2}) -\Erf(1-(R^2-p^2)^{1/2}) \right) .
  \end{split}
\end{equation}
For $R=1.6$, say, one gets a reasonable dark count rate of
$\langle 0|\Pi_1|0\rangle=0.0773$
and
$\langle 1|_c\Pi_0|1\rangle_c =0.7104$ 
, so an effective detector efficiency of 
$0.2896$.
These values are not that far off from those achieved in current experiments (see, e.g., Ref.~\cite{Detectors,Detectors2}).

In the setting considered here, the Wigner functions $W_{|0\rangle\langle 0|}$ and $W_\omega$ of the two single mode input states from which the initial state is constructed, that of the partial transposed of the Jamiolkowski isomorphs of all gates $W_{f^\Gamma_j},\ j\in [m^2]$, as well as that of the POVM elements $W_{\Pi_0}$ and $W_{\Pi_1}$ are non negative.
Therefore, the algorithms of Refs.~\cite{Mari,Emerson1} can be applied.
The detailed error analysis of Refs.~\cite{Emerson1,Emerson2} implies the following.

\begin{observation}[Efficient sampling in 1-norm for imperfect detectors] 
  For any number of modes $m$, any $R>0$, any Gaussian product input state $\rho$, in which each mode is prepared in either the vacuum $|0\rangle\langle 0|$ or an arbitrary Gaussian state $\omega$, any linear optical network, and any dichotomic detector with POVM elements $\Pi_0$ and $\Pi_1$ as defined in Eq.~\eqref{Wi}, one can sample from the output distribution over $S \in \Phi_{m,n}$, to an error $\epsilon$ in 1-norm with effort $\landauO(\poly(m/\epsilon))$.
\end{observation}

That is to say, one can efficiently simulate the output distribution of imperfect linear optical networks and imperfect detectors of the above type even up to a small error in 1-norm.
We suggest that it should be an important and constructive enterprise to exactly flesh out how far one can go with
approximating realistic experimental devices, while still being able to provably efficiently simulate the output distribution.

\section{Conclusion and outlook}
In this work, we have revisited the Boson-Sampling problem from the perspective of sample complexity.
We have arrived at the ironic conclusion that no symmetric probabilistic algorithm can distinguish the Boson-Sampling distribution from the mere uniform distribution on the collision-free subspace, unless exponentially many samples are available.
The specifics of the problem if a priori knowledge is available have been discussed carefully.
We have also addressed the question to what extend imperfect, approximate physical realisations of the Boson-Sampling problem can be classically efficiently simulated up to a constant error in 1-norm.
As such, our work emphasizes the challenge of identifying ways to certify the correct working of such quantum simulators.
Our results indicate that even though, unquestionably, the Boson-Sampling distribution has an intricate structure that makes sampling from it a classically hard problem, this structure seems inaccessible by classical means.
To develop a portfolio of methods for \emph{certifying the correct functioning of quantum simulators} seems timelier than ever.
Probably, quantum methods are indispensable to achieve that goal.

The question of the precise boundary of classically simulatable quantum processes remains wide open and interesting, and is also enjoying an increasing amount of attention \cite{Hoban13,Wiebe13}, not least because of the rather loud claims made in the context of the discussion on the functioning of the D-Wave processor and their careful assessment \cite{DW1,DW2,DW3}.
It is the hope that the present work can contribute to a thoughtful scientific reasoning on identifying the boundary of classically simulatable processes in general, and at the same time contribute to clarifying in what precise sense quantum devices such as Boson-Samplers are indeed more powerful than classical devices.

Obviously, technically, our argument leaves significant room for improvement.
It would be interesting to see, for example, whether, or to what extent, a priori knowledge on the distribution can be used or what other important features of the Boson-Sampling distribution may be identified.
It would also be important to see how the hardness argument can be partially \emph{de-randomised}, and the Haar-measure random unitaries replaced by appropriate \emph{unitary designs} or related concepts derived from \emph{quantum expanders}.

We also hope that our work can be read as yet another invitation to the enterprise of looking at the sample complexity of tasks in quantum theory.
Quite generally, all information that is ever available in any quantum mechanical experiment is obtained from samples from a certain distribution.
These samples may be used to infer about important features or properties of the underlying quantum state --- or even about the very identity of the state in the first place. The quantum state tomography problem --- the inference about an unknown quantum state from measurement data alone --- should be phrased as a sampling problem.
Indeed, the \emph{tomography problem} has already been faithfully viewed as a sampling problem and the sample complexity lower bounded, both in the context of \emph{quantum compressed sensing} \cite{Flammia} and in notions of \emph{reliable quantum state tomography} \cite{Christandl}.
A similar mindset has been taken in foundational arguments explaining the apparent emergence of ensembles of quantum  statistical mechanics based on microscopic unitary evolution \cite{Ududec}: 
Indeed, one may argue that if by sampling alone, one cannot operationally distinguish a situation from the one predicted by a statistical ensemble, then the apparent emergence may be considered explained.
It is the hope that the methods discussed in this work suggest further applications along these lines.

\section{Acknowledgments}
We warmly thank Fernando G.~S.~L.\ Brandao, Earl T.\ Campbell, and Rodrigo Gallego for insightful discussions and Scott Aaronson and Alex Arkhipov for useful criticism.
We thank the EU (Q-Essence, REQS Marie Curie IEF No 299141), the ERC (TAQ), the EURYI, the BMBF (QuOReP), and the Studienstiftung des Deutschen Volkes for support.

\end{document}